\documentclass[5p, twocolumn]{elsarticle}
\usepackage{graphicx}
\usepackage{listings}
\usepackage{verbatim}
\usepackage{fancybox}
\usepackage[ps2pdf, bookmarks=true]{hyperref}
\usepackage{array}
\usepackage{amsmath,amssymb}
\usepackage{psfrag}

\journal{Reconfigurable Computing}

\begin{document}

\begin{frontmatter}
\title{Cyclic Sequence Generators as Program Counters for High-Speed FPGA-based
Processors}
\author{P. A. Suggate}
\author{R. W. Ward}
\author{T. C. A. Molteno}
\ead{tim@physics.otago.ac.nz}
\address{Department of Physics, University of Otago, P.O. Box 56, Dunedin, New Zealand}

\begin{abstract}

This paper compares the performance of conventional radix-2 program counters with
program counters based on Feedback Shift Registers (FSRs), a class of cyclic
sequence generator. FSR counters have constant time scaling with bit-width, $N$,
whereas FPGA-based radix-2 counters typically have $O(N)$ time-complexity due to
the carry-chain. Program counter performance is measured by synthesis of
standalone counter circuits, as well as synthesis of three FPGA-based processor
designs modified to incorporate FSR program counters. Hybrid counters, combining
both an FSR and a radix-2 counter, are presented as a solution to the potential
cache-coherency issues of FSR program counters. Results show that high-speed
processor designs benefit more from FSR program counters, allowing both greater
operating frequency and the use of fewer logic resources.

\end{abstract}


\end{frontmatter}

\tableofcontents

\section{Introduction}

A Program Counter (PC) circuit~\cite{parhami2005cam} generates the address of the
next instruction to be fetched for execution. The greatest contributor to total
PC circuit latency in an FPGA-based processor can be due to the counter that is
used to increment the current PC value. This is because conventional radix-2
counters implemented within FPGAs can have long carry chains. Pipelining can be
used to obtain a higher operating frequency but also increases logic usage,
and likely increases the branch penalty.

Non-radix-2 cyclic sequence generators can be used to generate the next
instruction address, for example a maximum-cycle Feedback Shift Register
(FSR)~\cite{golomb1981srs, robshaw1995sc, menezes1997ac}, and can lead to a
reduction in total PC latency. This is because FSR counters can be designed where
the maximum depth of combinatorial logic required is only one
gate~\cite{high_performance_ring_generators} therefore the PC latency is constant
with bit-width  $N$. We explore how this reduction in PC circuit complexity can
affect maximum processor operating frequency for three FPGA processor designs.

The sequence of instruction addresses generated by a PC circuit using a
maximum-cycle FSR is pseudo-random. For a processor fetching instructions from a
small embedded memory this presents no problem. For processors that feature an
instruction cache an FSR PC will have poor cache-coherency behaviour. As as
solution we present a hybrid PC architecture. The hybid-PC is the concatenation
of two smaller counters. The hybrid PC uses a small radix-2 counter to step
through instructions within a cache line, and a FSR counter that cycles between
cache line. When implemented within FPGA processors, this hybrid approach has low
latency and avoids cache-coherency problems.

\section{Program Counters}
\label{PCs}

Stan et al.~\cite{stan1998laf} list the properties of generic up/down-counters
but not all of these properties are required for PC circuits. A program counter
must be \verb+RESET+able, increment once every clock cycle as long as an
\verb+ENABLE+ line is asserted, and sometimes has its value changed by branch
instructions (so it needs to be \verb+LOAD+able, and have \verb+IN+ lines). The
value also needs to be readable every clock cycle (using the \verb+OUT+ lines) in
order to access the memory address to fetch the next instruction. There is no
need to support other common counter features~\cite{stan1998laf}, such as being
reversible, or any terminal count operations. A black box diagram for a generic
PC circuit is shown in Figure~\ref{pc_black_box}.

\begin{figure}[ht!]
\begin{center}
\includegraphics[width=0.5\linewidth]{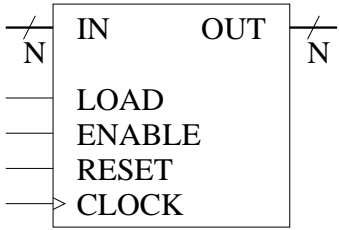}
\caption{Black box program counter model showing the required control signals.
The LOAD, ENABLE, and RESET actions all occur on the positive edge of the clock
(CLOCK).}
\label{pc_black_box}
\end{center}
\end{figure}

%
%
%
%
%

\section{Formal framework for counters}

In order to investigate structures that can be used for program counters, it is worth setting up some formalism.

A conventionial radix-2 counter can take on a range of values from 0 to $2^n-1$ where $n$ is the number of bits. The next value of the counter is obtained by incrementing the current value using radix-2 arithmetic.
We will show there is a isomorphism between this counter architecture and a family of other counter architectures. We consider only finite counters, as all counters implemented in digital logic will have a finite number of states.

We will start by defining a counter as finite and closed. Closed is not necessary for a program counter (all we need is a sufficent sequence of states to put the program in), but give us some nice properties.

We will then set up a `ordinary' counter $\mathbf{C}_n$ which is a cycle of length $n$. We will use isomorphism to this.

Lemma~\ref{lemma_cycle_full} is a useful grab bag of properties. Note that from Lemma~\ref{lemma_cycle_full}(4,5), any state can be used as a generator.

Lemma~\ref{lemma_iso} states that all $n$-cyclic counters are isomorphic to $\mathbf{C}_n$ (no suprises).

Corollary~\ref{corollary_cyclic_sub} states that all counters have a $n$-cyclic subcounter form some $n$. Just find the limit cycle. Note that the limit cycle may have only one element in it.

The central theorem is Theorem~\ref{counter_central_theorem}, that all counters have a sub-counter that is isomorphic to $\mathbf{C}_n$ for some $n$. This ties these new counters to counters with known behaviour.


Theorem~\ref{counter_concat} lets us join counters together to form hybrid counters.

\begin{figure}
\begin{center}
\psfrag{Possible Program Counters}{\Large Possible Program Counters}
\psfrag{Radix-2 Counters}{\Large Radix-2 Counters}
\psfrag{Maximum Cycle FSRs}{\Large Maximum Cycle FSRs}
\psfrag{PC=Cn}{\Large$\text{PC}\cong C_n$}
\psfrag{PC=C2m}{\Large$\text{PC}\cong C_{2^m}$}
\psfrag{PC=C2m-1}{\Large$\text{PC}\cong C_{2^m-1}$}
\includegraphics[width=\linewidth]{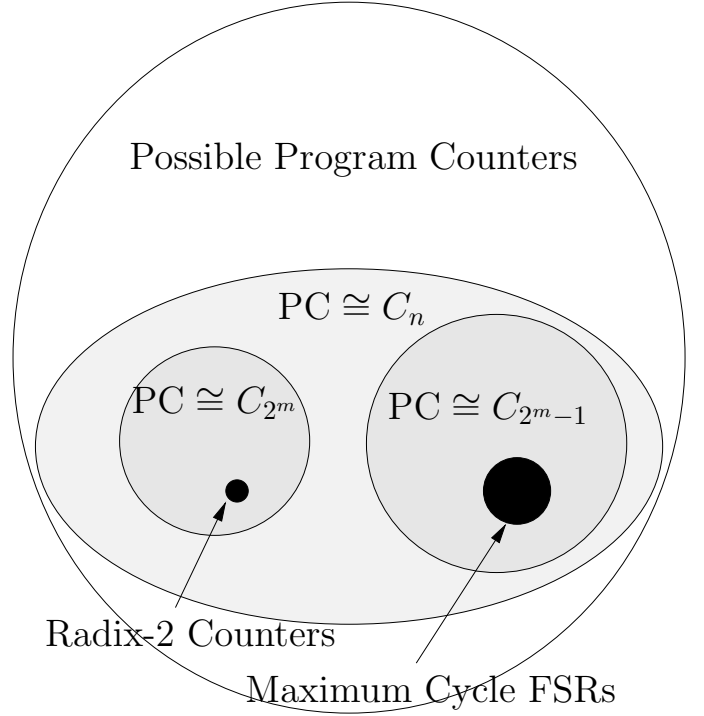}
\caption{Types of Program Counter}
\label{isomorphism}
\end{center}
\end{figure}

\newtheorem{theorem}{Theorem}[section]
\newtheorem{lemma}[theorem]{Lemma}
\newtheorem{proposition}[theorem]{Proposition}
\newtheorem{corollary}[theorem]{Corollary}

\newenvironment{proof}[1][Proof]{\begin{trivlist}
\item[\hskip \labelsep {\bfseries #1}]}{\end{trivlist}}
\newenvironment{definition}[1][Definition]{\begin{trivlist}
\item[\hskip \labelsep {\bfseries #1}]}{\end{trivlist}}
\newenvironment{example}[1][Example]{\begin{trivlist}
\item[\hskip \labelsep {\bfseries #1}]}{\end{trivlist}}
\newenvironment{remark}[1][Remark]{\begin{trivlist}
\item[\hskip \labelsep {\bfseries #1}]}{\end{trivlist}}

\begin{definition}
Consider a finite set of states $\mathbb{S}$ and an increment operator $\sigma:\mathbb{S}\mapsto \mathbb{S}$. For the purposes of this paper we define a counter as a pair $(\mathbb{S},\sigma)$ with the following property:

\begin{align}
& \forall s \in \mathbb{S}, \sigma(s) \in \mathbb{S}
\label{completeness_axiom}
\end{align}
\end{definition}

This (closure) is one of Peano's axioms for the natural numbers~\cite{peano} -- we don't need the other four axioms for our purposes.

A useful special case is a modulo arithmetic counter, with the ordinary increment operator.

\begin{definition}
Let $\mathbf{C}_n$ be the counter $(\mathbb{Z}_n,f_n)$ where $f_n(x)=x+1 \text{ mod } n$.
\end{definition}

The conventional radix-2 counter with $m$ bits is now simply $\mathbf{C}_{2^m}$.

\begin{definition}
We also define $\sigma^n(s)$ as $n$ successive applications of $\sigma$, so
\begin{align}
& \sigma^0(s)=s \\
& \forall n>0, \sigma^n(s)=\sigma(\sigma^{n-1}(s))
\end{align}
\end{definition}

\begin{definition}
Given a counter $(\mathbb{S},\sigma)$, $(\mathbb{T},\sigma)$ is a subcounter of $(\mathbb{S},\sigma)$ if $(\mathbb{T},\sigma)$ is a counter and $\mathbb{T}\subseteq\mathbb{S}$.
\end{definition}

\begin{definition}
A counter $(\mathbb{S},\sigma)$ is $n$-cyclic if $|\mathbb{S}|=n$, $\exists s_0\in\mathbb{S}, \mathbb{S}=\{\sigma^i(s_0)|i\geq 0\}$, and $\exists m>0, \sigma^m(s_0)=s_0$. $s_0$ is a \emph{generator} for the counter.
\end{definition}

\begin{lemma}
Given an $n$-cyclic counter $(\mathbb{S},\sigma)$, then for any element $s\in\mathbb{S}$
\begin{enumerate}
\item $\forall t\in\mathbb{S}, \exists i\geq0, \sigma^i(s)=t$,
\item $\sigma^i(s)\neq s$ for $0<i<n$,
\item $\forall i,j, 0\leq i,j<n, i\neq j\Rightarrow \sigma^i(s)\neq\sigma^j(s)$,
\item $\{\sigma^i(s)|0\leq i<n\}=\{\sigma^i(s)|i\geq 0\}=\mathbb{S}$,
\item $\sigma^n(s)=s$.
\end{enumerate}
\label{lemma_cycle_full}
\end{lemma}

\begin{proof}
For (1), $s=\sigma^j(s_0)$ and  $t=\sigma^k(s_0)$. If $j\leq k$, $t=\sigma^k(s_0)=\sigma^{k-j}(\sigma^j(s_0))=\sigma^{k-j}(s)$. If $j>k$, we know that 
$\exists m>0, \sigma^m(s_0)=s_0$, so choose $x$ such that $mx>j-k$. $\sigma^{mx}(s_0)=s_0$, so $t=\sigma^k(s_0)=\sigma^{k+mx}(s_0)=\sigma^{k+mx-j}(\sigma^j(s_0))=\sigma^{k+mx-j}(s)$.

For (2), assume the contrary. If $\sigma^i(s)=s$ for some $0<i<n$, then consider the set of states $\mathbb{T}=\{\sigma^j(s)|0\leq j<i\}$. For any $t\in \mathbb{T}$, $t=\sigma^j(s)$ for some $0\leq j<i$. If $j<i-1$, then $\sigma(t)=\sigma(\sigma^j(s))=\sigma^{j+1}(s)\in\mathbb{T}$. If $j=i-1$, $\sigma(t)=\sigma(\sigma^j(s))=\sigma^i(s)=s\in\mathbb{T}$. Hence $\mathbb{T}$ is closed under $\sigma$. However $|\mathbb{S}|=n$ and $|\mathbb{T}|=i<n$, so $\exists u\in\mathbb{S}\setminus\mathbb{T}$. However, $\mathbb{T}$ is closed under $\sigma$, so $\nexists p\geq 0, \sigma^p(s)=u$, contradicting (1). Hence $\sigma^m(s)\neq s$ for $0<i<n$.

For (3), If $\sigma^i(s)=\sigma^j(s)$ for some $0\leq i<j<n$ (assuming w.l.o.g. that $i<j$), then $\sigma^i(s)=\sigma^j(s)=\sigma^{j-i}(\sigma^i(s))$ where $j-i<n$, which violates (2).

For (4), $\sigma^i(s)$ where $0\leq i<n$ are elements of $\mathbb{S}$ and from (3) these values are distinct, so $|\{\sigma^i(s)|0\leq i<n\}|=n=|\mathbb{S}|$. $\exists j, s=\sigma^j(s_0)$ so $\{\sigma^i(s)|0\leq i<n\}=\{\sigma^i(\sigma^j(s_0))|0\leq i<n\}=\{\sigma^i(s_0)|j\leq i<n+j\}\subseteq\mathbb{S}$ (since \{$\mathbb{S},\sigma\}$ is closed) hence $\{\sigma^i(s)|0\leq i<n\}=\mathbb{S}$. Also, $\mathbb{S}=\{\sigma^i(s)|0\leq i<n\}\subseteq|\{\sigma^i(s)|0\leq i\}=|\{\sigma^i(s_0)|j\leq i\}\subseteq\mathbb{S}$, so $\{\sigma^i(s)|i\geq 0\}=\mathbb{S}$.

For (5), we use (4) to give $\sigma^n(s)=\sigma^i(s)$ for some $0\leq i<n$. Now $\sigma^n(s)=\sigma^{n-i}(\sigma^i(s))=\sigma^{n-i}(\sigma^n(s))$. If $n-i<n$ then (2) is contradicted, and $i \geq 0$, so $n-i=n$ and $i=0$. Hence $\sigma^n(s)=\sigma^0(s)=s$.
\end{proof}

\begin{lemma}
$\mathbf{C}_n$ is $n$-cyclic.
\label{cn_cyclic}
\end{lemma}

\begin{proof}
Just let $s_0=0$. $|\mathbb{Z}_n|=n$, $\mathbb{Z}_n=\{f_n^i(x)|i\geq 0\}$ and $f_n^n(0)=0$.
\end{proof}

\begin{lemma}
If $(\mathbb{S},\sigma)\cong(\mathbb{T},\tau)$, then $(\mathbb{S},\sigma)$ is $n$-cyclic iff $(\mathbb{T},\tau)$ is $n$-cyclic.
\label{iso_cyclic}
\end{lemma}

\begin{proof}
Assume $(\mathbb{S},\sigma)$ is $n$-cyclic. $(\mathbb{S},\sigma)\cong(\mathbb{T},\tau)$ so there exists a mapping $g:\mathbb{S}\mapsto\mathbb{T}$ such that $\forall s\in\mathbb{S},g(\sigma(\mathbb{s}))=\tau(g(\mathbb{s}))$ that is bijective. If $i\geq0$, $g(\sigma^{i+1}(\mathbb{S}))=g(\sigma(\sigma^i(\mathbb{S})))=\tau(g(\sigma^i(\mathbb{S})))$, hence it can easily be seen by induction that  $\forall i\geq0,g(\sigma^i(\mathbb{s}))=\tau^i(g(\mathbb{s}))$.

$|\mathbb{T}|=|\mathbb{S}|=n$ (from the bijection of $g$). Given $s_0\in\mathbb{S}, \mathbb{S}=\{\sigma^i(s_0)|i\geq 0\}$, then $\{\tau^i(g(s_0))|i\geq 0\}=\{g(\sigma(s_0))|i\geq 0\}=\{g{s}|s\in\mathbb{S}\}=\mathbb{T}$, since $g$ is bijective. $\tau^m(g(s_0))=g(\sigma^m(s_0))=g(s_0)$ (from the definition of $n$-cyclic), so $(\mathbb{T},\tau)$.

If $(\mathbb{T},\tau)$ is $n$-cyclic, the proof that $(\mathbb{S},\sigma)$ is $n$-cyclic follows from using the same argument with $g^{-1}$.
\end{proof}

\begin{lemma}
A counter is $n$-cyclic iff it is isomorphic to $\mathbf{C}_n$. In particular, if $(\mathbb{S},\sigma)$ is $n$-cyclic, $g:\mathbb{Z}_n\mapsto\mathbb{S}$ as $g(i)=\sigma^i(s_0)$ is an isomorphism for any $s_0\in\mathbb{S}$.
\label{lemma_iso}
\end{lemma}

\begin{figure}
\begin{center}
\psfrag{S}{$\mathbb{Z}_n$}
\psfrag{T}{$\mathbb{S}$}
\psfrag{s}{$\!\!f_n$}
\psfrag{t}{$\sigma$}
\psfrag{g}{g}
\psfrag{s0}{$j$}
\includegraphics[width=0.5\linewidth]{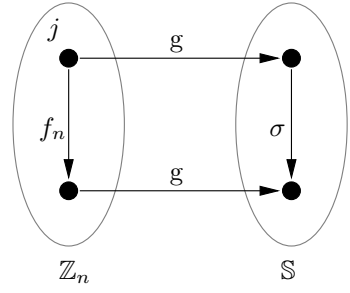}
\caption{Homomorphism between $C_n$ and $(\mathbb{S},\sigma)$ as used in the proof for Lemma~\ref{lemma_iso}}
\label{homomorphism}
\end{center}
\end{figure}

\begin{proof}
Firstly we show that of a counter $n$-cyclic, then it is is isomorphic to $\mathbf{C}_n$. We construct an a mapping $g$ and show it is an isomorphism. Given an $n$-cyclic counter $(\mathbb{S},\sigma)$, pick any element $s_0\in\mathbb{S}$. Define the mapping $g:\mathbb{Z}_n\mapsto\mathbb{S}$ as $g(i)=\sigma^i(s_0)$. Given any $j$ such that $0\leq j<n$, there are two cases. If $j<n-1$ then $g(f_n(j))=g(j+1)=\sigma^{j+1}(s_0)=\sigma(\sigma^j(s_0))=\sigma(g(j))$. If $j=n-1$ then $g(f_n(j))=g(0)=s_0=\sigma^n(s_0) \text{ (using lemma~\ref{lemma_cycle_full}(5)) }=\sigma(\sigma^{n-1}(s_0))=\sigma(g(j))$. Hence $g$ is a homomorphism as shown in Figure~\ref{homomorphism}. From lemma~\ref{lemma_cycle_full}(3) $g$ is an injection, and from lemma~\ref{lemma_cycle_full}(4) $g$ is a surjection, hence $g$ is an isomorphism, and $(\mathbb{S},\sigma)\cong\mathbf{C}_n$.

Going the other way, if a counter is isomorphic to $\mathbf{C}_n$, it is $n$-cyclic from lemmata~\ref{cn_cyclic} and \ref{iso_cyclic}.
\end{proof}

\begin{theorem}[Generator]
Given a counter $(\mathbb{S},\sigma)$, $s\in\mathbb{S}$ and an integer $n$ such that $\sigma^m(s)\neq s$ for $0<m<n$ and $\sigma^n(s)=s$, then $(\{\sigma^i(s)|0\leq i<n\},\sigma)$ is an $n$-cyclic subcounter of $(\mathbb{S},\sigma)$.
\label{theorem_counter_generator}
\end{theorem}

\begin{proof}
Let $\mathbb{T}=\{\sigma^i(s)|0\leq i<n\}$. First we show that $(\mathbb{T},\sigma)$, $s\in\mathbb{S}$ is a subcounter of $(\mathbb{S},\sigma)$. Consider any element $t\in\mathbb{T}$. $t=\sigma^k(s_0)$ for some $0\leq k<n$. $\sigma(t)=\sigma(\sigma^k(s))=\sigma^{k+1}(s)$. If $k<n-1$ then $k+1<n$, so $\sigma(s)\in\mathbb{T}$. If $k=n-1$, $\sigma(t)=\sigma^n(s)=\sigma(s)=\sigma^0(s)\in\mathbb{T}$, so equation~\ref{completeness_axiom} is satisfied, making $\mathbb{T}$ a counter. It is a subcounter of $(\mathbb{S},\sigma)$ because $\mathbb{T}\subseteq\mathbb{S}$.

We need to show that the cardinality of $|\mathbb{T}|=n$. Clearly $|\mathbb{T}|\leq n$ by the construction. If $|\mathbb{T}|<n$, then there would be at least one repeated state, so $\exists 0\leq k,l<n, k\neq l, \sigma^k(s)=\sigma(l)$. Assume w.l.o.g that $k<l$. Now $s=\sigma^n(s)=\sigma^{n-l}(\sigma^l(s))=\sigma^{n-l}(\sigma^k(s))=\sigma^{n+k-l}s$. $0<n+k-l<n$ which $\sigma^m(s)\neq s$ for $0<m<n$, so  $|\mathbb{T}|=n$.

Consider any $\sigma^i(s),i\geq 0$. Let $j=i\text{ mod }n$, so $\exists x,i=j+nx$. Now $\sigma^i(s)=\sigma^{i-nx}\sigma^{nx}(s)=\sigma^j(s)$. So $\{\sigma^i(s)|0\leq i<n\}=\{\sigma^i(s)|i\geq 0\}$.

Putting this together, $|\mathbb{T}|=n$, $\exists s\in\mathbb{T}, \mathbb{T}=\{\sigma^i(s)|i\geq 0\}$, and $\sigma^n(s)=s$, so $(\mathbb{T},\sigma)$ is an $n$-cyclic subcounter of $(\mathbb{S},\sigma)$.

\end{proof}

\begin{remark}
It is worth noting (although we do not use or prove the result) that any element $s$ of an $n$-cyclic counter $(\mathbb{S},\sigma)$can be used to generate the counter in this way.
\end{remark}

\begin{corollary}
All counters have an $n$-cyclic subcounter for some $n$.
\label{corollary_cyclic_sub}
\end{corollary}

\begin{proof}
This can be done by construction. Consider any counter $(\mathbb{S},\sigma)$. $\mathbb{S}$ is defined to be finite, so let $m$ be the cardinality of $\mathbb{S}$. Now pick any element $s_0\in\mathbb{S}$ and define a sequence $<\sigma^0(s_0),\sigma^1(s_0),\sigma^2(s_0), \ldots ,\sigma^{m-1}(s_0),\sigma^m(s_0)>$. Because of equation~\ref{completeness_axiom}, all members of this sequence are elements of $\mathbb{S}$, and the sequence has $m+1$ members. Therefore there must be at least one value that occurs more than once in this sequence, so we can find two integers $i$ and $j$ such that $\sigma^i(s_0)=\sigma^j(s_0)$ where w.l.o.g. $i<j$ and for all $k$ such that $i<k<j, \sigma^j(s_0)\neq\sigma^k(s_0)$ (if there was such a $k$ we could choose $k$ instead of $j$ for the second value).

Let $\mathbb{T}=\{\sigma^k(s_0)|i\leq k<j\}$. $\{\sigma^k(s_0)|i\leq k<j\}=\{\sigma^k(\sigma^i(s_0))|0\leq k<j-i\}$, so by theorem\ref{theorem_counter_generator}, $(\mathbb{T},\sigma)$ is an $n$-cyclic subcounter of $(\mathbb{S},\sigma)$ where $n=j-i$.
\end{proof}

\begin{theorem}
All counters have a sub-counter that is isomorphic to $\mathbf{C}_n$ for some $n$.
\label{counter_central_theorem}
\end{theorem}

\begin{proof}
This follows directly from lemma~\ref{lemma_iso} and corollary~\ref{corollary_cyclic_sub}.
\end{proof}

%

\begin{definition}
Given two counters  $(\mathbb{T},\tau)$, $(\mathbb{S},\sigma)$ and an element $s_0\in \mathbb{S}$, we define $(\mathbb{T},\tau)\oplus_{s_0}(\mathbb{S},\sigma)$ as the structure $(\mathbb{T}\otimes\mathbb{S},\tau\oplus_{s_0}\sigma)$ where

\begin{equation}
\tau\oplus_{s_0}\sigma((t,s))=
\begin{cases}
(\tau(t),\sigma(s)) & s=s_0 \\
(t,\sigma(s)) & s\neq s_0 \\
\end{cases}
\end{equation}
\end{definition}

$(\mathbb{T},\tau)\oplus_{s_0}(\mathbb{S},\sigma)$ is clearly a counter, as equation~\ref{completeness_axiom} is satisfied trivially.

\begin{theorem}
If $(\mathbb{T},\tau)$ is an $n$-cyclic counter and $(\mathbb{S},\sigma)$ is an $m$-cyclic counter, then for any $(\mathbb{T},\tau)\oplus_{s_0}(\mathbb{S},\sigma)$ is an $mn$-cyclic counter.
\label{counter_concat}
\end{theorem}

\begin{proof}
We construct a mapping in $g:\mathbf{C}_m\mapsto\mathbb{S}$ as $g(i)=\sigma^i(s_0)$ and $h:\mathbf{C}_n\mapsto\mathbb{T}$ as $h(i)=\tau^i(t_0)$ for some $t_0\in\mathbb{T}$. These are isomorphisms from lemma~\ref{lemma_iso}. Now consider $\mathbf{C}_n\oplus_{m-1}\mathbf{C}_m$, and a map $p:\mathbb{Z}_n\otimes\mathbb{Z}_m\mapsto\mathbf{Z}_{mn}$ defined as $p((j,i))=jm+i$. $p((f_n\oplus_{m-1} f_m)((j,m-1))=p((j+1\text{ mod }n,0))=(j+1\text{ mod }n)m=jm+m\text{ mod }mn$, and $\forall 0\leq i<m-1, p((f_n\oplus_{m-1} f_m)((j,i))=p((j,i+1))=jm+i+1=jm+i+1\text{ mod }mn$ (since $jm+i+1<mn$), so $\forall 0\leq i<m, p((f_n\oplus_{m-1}f_m)((j,i))=p((j,i+1))=jm+i+1\text{ mod }mn=p((j,i))+1\text{ mod }mn$. Hence $p$ is a homomorphism from $\mathbf{C}_n\oplus_{m-1}\mathbf{C}_m$ to $\mathbf{C}_{mn}$. It is also clearly a bijection and a surjection, hence an isomorphism. ... 
\end{proof}

There is one operation other than increment that may be performed on a Program Counter - the \verb+LOAD+ operation (\verb+RESET+ is simply a special case of this). There are two forms of \verb+LOAD+ that are currently used:

A program might perform an absolute jump (or \verb+LOAD+). This simply loads the counter with a new value that has been calculated when the program is created, and is of no difficulty for any sort of counter.

The other form, which is worth examining in some detail, is where the new address is calculated as some offset $b$ from the current one by `adding' it to the current value of the program counter. Our system doesn't have addition defined, so we define one:

\begin{definition}
Given an counter $(\mathbb{S},\sigma)$, $s\in\mathbb{S}$ and an integer $b\geq 0$ then $s+b=\sigma^b(s)$.
\end{definition}

In the case of $\mathbf{C}_n$, $s+b=f_n^b(s)=s\text{ `plus' }b\text{ mod }n$, where `plus' is the usual defintion of modulo addition.

\section{Radix-2 Counters}
\label{radix2_counters}

Synchronous radix-2 counters are conventionally used for Program Counters,
however there is a trade-off between speed and size because of carry propagation
from low-order to higher-order bits~\cite{stan1998laf}.

The simplest radix-2 counter is a ripple-carry counter based on an
adder~\cite{stan1998laf}. This is slow ($O(N)$ combinatorial delay with
increasing bit-width $N$) but cheap in logic ($O(N)$). Xilinx FPGAs use
carry-chains to provide the logic for this. This can be improved to $O(\log N)$
combinatorial delay using a carry-lookahead design~\cite{parhami2000computer} at
the expense of extra logic.

An approach that allows increment in constant time is a redundant
format~\cite{parhami-systolic} which allows what is called `carry-free' addition,
but that has two issues. A redundant representation needs twice as many latches
and therefore consumes more FPGA resources. Another problem is that a redundant
output is unsuitable for providing an address to access instructions. The output
first needs to be converted to a non-redundant format during each clock cycle,
which is going to require $O(\log N)$ propagation delay and extra logic for an
adder. Some work has been done using hybrid redundant number
systems~\cite{phatak1994hybrid} to reduce the number of extra registers, but
there is still a substantial propagation delay converting that to a usable
format.

Another approach to designing an $O(1)$ delay counter is to use a
cascade~\cite{parhami2000computer} that begins with a short and fast counter, and
continues with longer counters that only need to be incremented occasionally and
don't need to be as fast. However, this requires the slower counters to have
their increments to be precomputed~\cite{stan1998laf} which makes the \verb+LOAD+
operation much more difficult.

\subsection{Cyclic Sequence Generators}


A cyclic sequence generator is a synchronous circuit that iterates through a
cyclic sequence of states. For a program counter, it is desirable to have a
simple circuit and a cyclic sequence that includes most of the possible states. A
good candidate for this is a maximum cycle feedback shift register.

A Feedback Shift Register (FSR) is a shift register that satisfies the condition
that the current state is generated by a linear function of its previous state.
There are many types of these, but the ones considered here are linear (use XOR
gates), and have a constant maximum combinatorial path, therefore constant time
performance ($O(1)$) with increasing bit width $N$ (see Figure~\ref{fsr8_figure}
for some examples of these). This compares favourably with radix-2 counters as
described in Section~\ref{radix2_counters}. A well-known example of an FSR is a
Linear Feedback Shift Register~\cite{golomb1981srs,menezes1997ac} (LFSR) that are
widely used in cryptography~\cite{menezes1997ac}, communications
systems~\cite{pickholtz1982tss} and for built-in self-test
systems~\cite{agrawal1993tbs}. A few of the types of FSR are:

\newcolumntype{V}{>{\centering\arraybackslash} m{.85\linewidth} }

\begin{figure}
\begin{center}
\begin{tabular}{c V}
(a) & \includegraphics[width=0.9\linewidth]{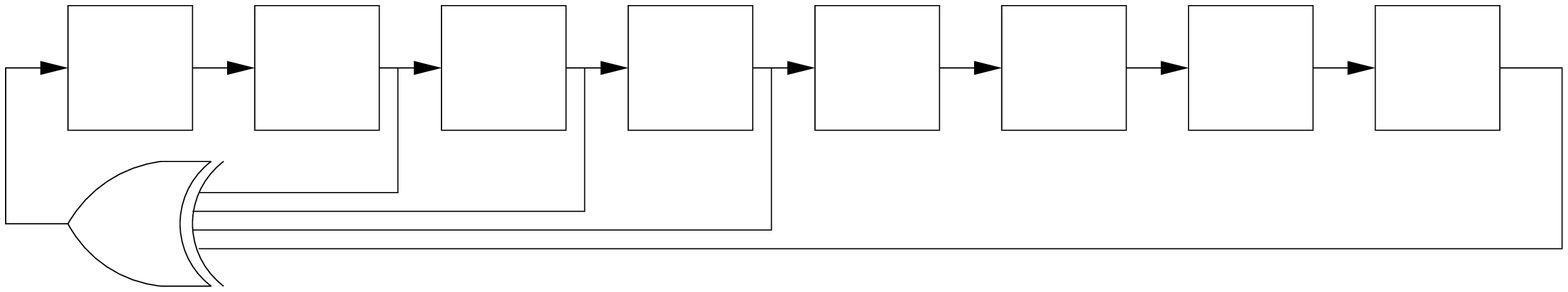} \\
(b) & \includegraphics[width=\linewidth]{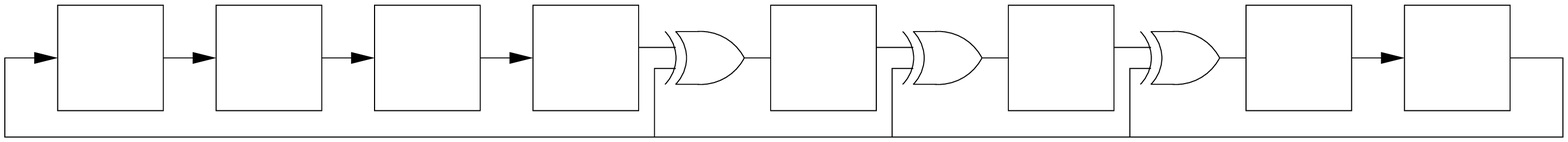} \\
(c) & \includegraphics[width=\linewidth]{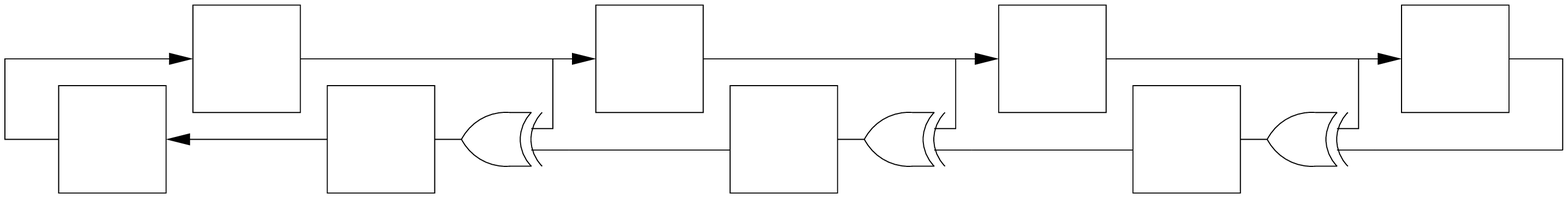} \\
(d) & \includegraphics[width=\linewidth]{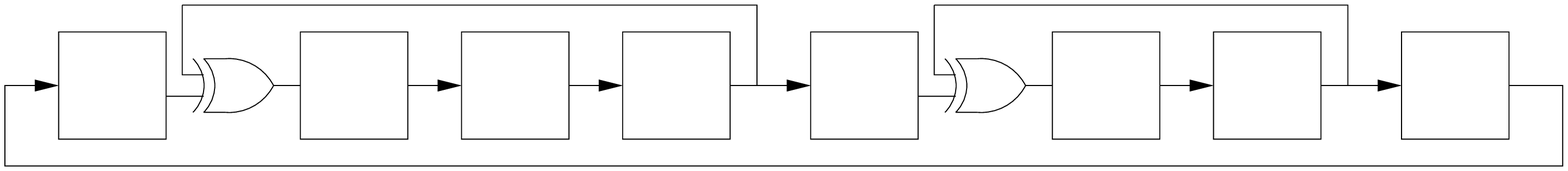} \\
\end{tabular}
\caption{8 bit counters with a cycle size of 255. (a) Fibonacci LFSR, (b)
Galois LFSR, (c) Ring Generator, (d) MFSR}
\label{fsr8_figure}
\end{center}
\end{figure}

\begin{description}
\item[Fibonacci LFSR] takes the output of several of the registers and XORs them together to feed to the input of the first register as shown in Figure~\ref{fsr8_figure}a.
\item[Galois LFSR] takes the output of the last register and XORs it with several of the register inputs as shown in Figure~\ref{fsr8_figure}b. 
\item[Ring Generators~\cite{high_performance_ring_generators}] rearrange the shift register into a ring, and arranges the feedback connections so that they only involve a small amount of routing, as shown in Figure~\ref{fsr8_figure}c.
\item[MFSR~\cite{MFSR_List}] (Multiple Feedback Shift Register) are a generalisation of ring generators, allowing any output to be XORed with any input, but limits the fan-in and fan-out to 2, as shown in Figure~\ref{fsr8_figure}d.
\end{description}

Also worth considering are cellular automata~\cite{cellular_automata}. These are
not Feedback Shift Registers but have similar properties. A cellular automaton
has each bit set from the XOR of the previous value of that bit and neighbouring
bits, which require more XOR gates than FSRs but keeps the routing very local.

An $N$-stage linear FSR is {\em maximum-cycle} when all  $2^N\!-\!1$ non-zero
states occur as the FSR is iterated. Note that a cycle of $2^N$ cannot be
achieved as the all zero state will always map to the all zero state. Any linear
FSR can be represented as an $N\times N$ matrix $\mathbf{M}$ over the field
$GF(2)$, and this will be maximum cycle if and only of the characteristic
polynomial $p(x)=|M-xI|$ is primitive. Lists of primitive characteristic
polynomials and counters with maximum cycles can easily be
found~\cite{menezes1997ac,xilinx052,ward_mfsr_table09} or generated.

All these structures are very fast, as described above, and have a relatively
small amount of logic ($O(N)$ which will mostly be the registers to store the
bits), and a pseudo-random sequence. They will have similar performance, due to
all having a maximum combinatorial part of just one XOR gate (In Fibonacci this
may be 4 gates, but this is still only one LUT on an FPGA). Slight variations in
performance may depend on required LUTs and routing in any particular FPGA
technology choice and application. The worst case number of XOR gates, fan in and
fan out is shown in Table~\ref{fsr_table}. As the sequence is pseudo-random, the
bits may also reordered to search for small routing improvements.

FSR counters meet the counter requirements for PC circuit described in
Section~\ref{PCs}. They are easily loadable, support an enable, and the register
can be read directly. One disadvantage to using FSRs is that the maximal cycle
size is $2^n-1$, instead of the $2^n$ cycle of radix-2 counters. We refer to the
address not generated in the $2^n-1$ cycle as the ``zero address''. While the
zero address does not represent a significant fraction of the address space, this
can be addressed with extra logic as used by Wang and
McCluskey~\cite{wang1988hdg}, but this brings the propagation delay back to
$O(\log N)$.

For the rest of this work we use MFSR counters. They have good fan-in, fan-out,
and a low gate count. Other FSRs could be used and would have very similar
results (see Figure~\ref{fsr8_figure}).

\begin{table}
\begin{center}
\begin{tabular}{r|ccc}
Structure & XOR gates & Fan-in & Fan-out \\
\hline
Fibonacci LFSR & 1 & 4 & 2 \\
Galois LFSR & 3 & 2 & 4 \\
Ring Generator & 3 & 2 & 2 \\
MFSR & 2 & 2 & 2 \\
Cellular Automata & $N$ & 3 & 3 \\
\end{tabular}
\caption{Worst case XOR gate count, Fan-in and Fan-out for several types of
Feedback Shift Register with $N$ bits}
\label{fsr_table}
\end{center}
\end{table}

\subsection{Hybrid PCs}
\label{PC_HYBRID}

The instruction fetch order of FSR PCs may lead to poor run-time performance for
processor designs that contain an instruction cache. It is desirable to have a PC
that increments through all of the instructions within a cache line before
fetching a new cache line. For a cache line-size of 32~bytes, and a fixed
instruction-width of 32~bits, there will be eight instructions within a cache
line. Ideally, and in the absence of branching instructions, each of these eight
instructions should be fetched from the cache before fetching another cache line
from system memory.

The solution presented here is to combine two counters into one PC, a radix-2
counter for the three least-significant bits and a MFSR for the most-significant
bits. Since the Spartan-3 contains four-input Look-Up Tables (LUTs), and the
Virtex-5 has six input LUTs, a 3-bit radix-2 counter can be built with just one
layer of logic. When the upper count value is reached, the MFSR portion of the PC
is then incremented.

\section{FPGA Synthesis of Simple Counter Circuits}

Performance of the synthesised radix-2 and FSR counters are shown in
Figure~\ref{GRAPH_COUNTERS}. As the bit-width $N$ increases radix-2 counters show
a linear increase in latency, this is $O(N)$ time-complexity. The results show
that for radix-2 counters of more than 6~bits, latency can be estimated as
$2.9+0.064 \times N$~ns. The FSR counters have $O(1)$ time complexity and a
smaller constant of only $1.8$~ns, compared with the radix-2 counters.

\begin{figure}[ht!]
\begin{center}
\includegraphics[width=\linewidth]{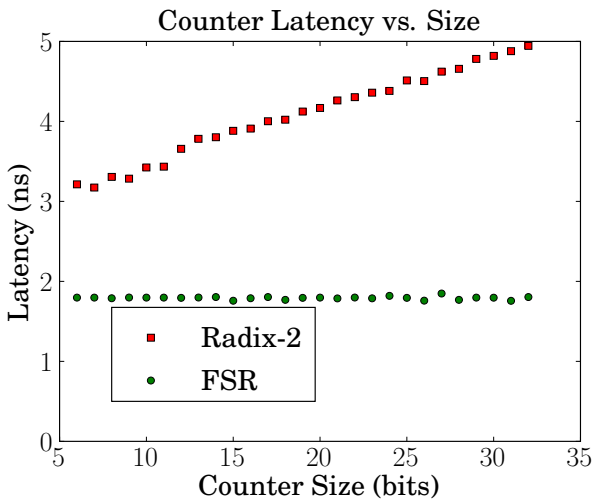}
\caption[Counter performance vs. counter size]{Counter
circuit performance vs. counter size when synthesised for a Spartan-3 FPGA.
The radix-2 counter implementation used was the default generated by the
Xilinx ISE 9.2 tools. This implementation uses the hardware carry-chain
present in Spartan-3 devices.}
\label{GRAPH_COUNTERS}
\end{center}
\end{figure}

%
\section{FPGA Synthesis of Complete PC Circuits}

A simple PC circuit was used when investigating the effect of counter-type on PC
performance. Figure~\ref{PC_CIRCUIT} is a block diagram of the circuit used for
testing and the counter used was one of FSR, radix-2, or a hybrid where the
lowest 3~bits are radix-2. Bit widths ranging from 8 to 32~bits were used.

\begin{figure}[ht!]
\begin{center}
\includegraphics[width=\linewidth]{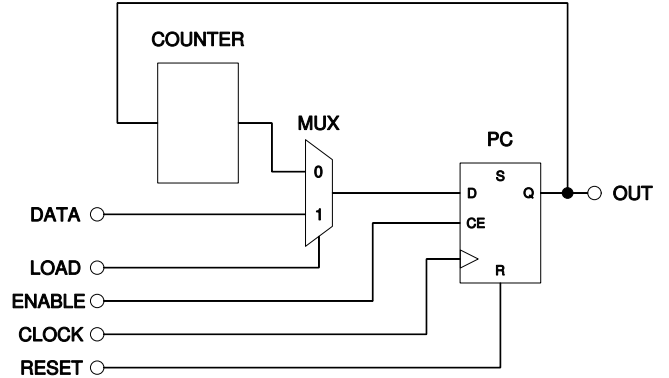}
\caption[Counter performance vs. counter size]{Block diagram of the PC
circuit used for the counter comparison tests. The counter type was either a
FSR, radix-2, or a FSR-radix-2 hybrid.}
\label{PC_CIRCUIT}
\end{center}
\end{figure}

\subsection{PC Circuit Results}

The results of the PC circuit synthesis are shown in Figure~\ref{GRAPH_PC}. As
with the previous counter tests, the radix-2 circuits scale linearly with
increasing bit-width. Again though, the latencies of FSR-based circuits are lower
than radix-2 and both the FSR and hybrid PCs have $O(1)$ scaling with increasing
bit-width. Due to the Xilinx synthesiser using the carry-chain logic for the
radix-2 counters, and FSRs having a gate depth of just one, this is as expected.

To demonstrate that the behaviour observed when synthesising for a Spartan-3 FPGA
is not unique, Figure~\ref{GRAPH_PC_VIRTEX} also shows synthesis results for a
Xilinx Virtex-5 FPGA. The latency scaling with increasing bit-width is similar
but the Virtex-5 is clearly a faster architecture.

\begin{figure}[ht!]
\begin{center}
\includegraphics[width=\linewidth]{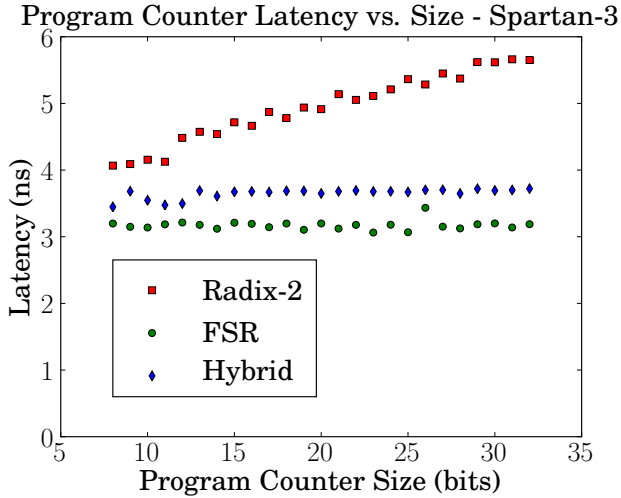}
\caption[Spartan-3 synthesis results of PC performance vs. counter size]{Xilinx
Spartan-3 FPGA synthesis results of PC circuit performance vs.
counter size. PCs with three counter types are compared: FSR, radix-2, and a
FSR-radix-2 hybrid (Lower 3~bits radix-2, upper bits FSR).}
\label{GRAPH_PC}
\end{center}
\end{figure}

\begin{figure}[ht!]
\begin{center}
\includegraphics[width=\linewidth]{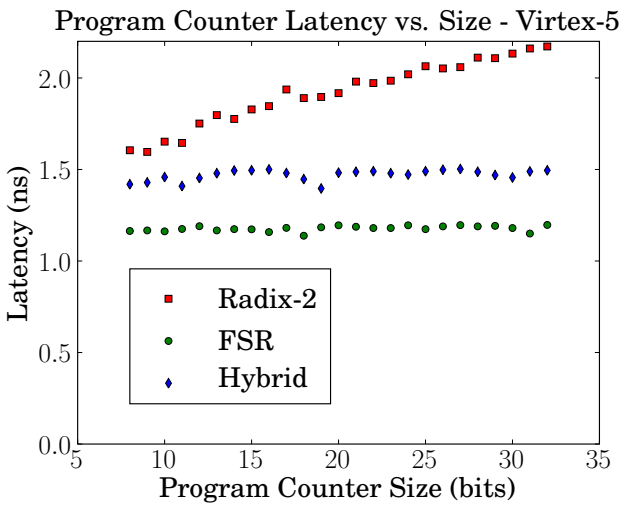}
\caption[Hybrid PC performance vs. counter size]{Xilinx Virtex-5 FPGA
synthesis results show that these devices are faster than Spartan-3 FPGAs.
The latency is lower but the comparative differences between counter types is
still very similar. The Virtex-5 FPGA is from the Xilinx performance range of
FPGAs. The Virtex-5 was used to demonstrate that the Spartan-3 results can
apply to other FPGAs.}
\label{GRAPH_PC_VIRTEX}
\end{center}
\end{figure}

%

\section{Processor Design Examples}

Three FPGA processors were synthesised and evaluated to show the effects of
different PC circuits on maximum clock frequency. The three processor logic cores
used are aeMB, TTA16, and RISC16. aeMB was designed to use a conventional radix-2
PC whereas RISC16 and TTA16 were designed to use a PC circuit based on either a
FSR or radix-2 counter.

This is to examine if substituting a FSR-based PC into an existing processor
design leads to any performance gains of the processor as a whole, and if there
is a detrimental effect of adding a radix-2-based PC to a processor designed for
a FSR-based PC. A Xilinx Spartan-3 FPGA was again the synthesis target and Xilinx
ISE 9.2 was the synthesis tool.

\subsection[aeMB]{aeMB: A MicroBlaze Compatible RISC Processor}

The aeMB processor logic core is 32-bit, Harvard architecture, RISC processor
with a three-stage pipeline. It is an open source project that was designed to be
instruction compatible with the MicroBlaze core~\cite{xilinx2008mpr}. When the
aeMB processor core was synthesised for a Spartan-3 FPGA it used about 2600
logic elements.

The critical path of aeMB, as determined from the Xilinx place-and-route timing
report, is the instruction-cache look-up path. Total routing resources used were
515643 paths with the original radix-2 PC. The source code was modified,
substituting a FSR PC for the radix-2 PC, and resynthesised. This was a straight
PC-logic substitution and did not address any potential problems with relative
branching and cache coherency.

Maximum operating frequency increased slightly, as is shown in
Table~\ref{CPU_Table}, logic resource utilisation was similar, but the FSR degign
used only 503608 paths. This is a large change in pathing resources used for a
small change to the total degign. This large difference in pathing resources,
caused only by the change to the PC circuit, was not observed with the other
processor cores.

\begin{table}[ht!]
\begin{center}
\begin{tabular}{l | l l l l}
CPU	&	PC Size	& FSR	& Radix-2	& Pipeline	\\
(Name)	&	(Bits)	& (MHz)	& (MHz)		& Stages	\\
\hline
TTA16	&	10	&	192	&	157	&	3	\\
RISC16	&	10	&	143	&	141	&	5	\\
aeMB	&	30	&	84	&	80	&	3	\\	
\end{tabular}
\caption[Processor performance with a FSR PC vs. radix-2 PC]{Three
different processor designs were synthesised, each with a FSR PC and with a
radix-2 PC. Gains are negligible except when the PC circuit is the critical
path, as with TTA16.}
\label{CPU_Table}
\end{center}
\end{table}


\subsection{TTA16}

TTA16 is a 16-bit, Transport Triggered Architecture~\cite{corporaal:tta} (TTA)
processor optimised for Xilinx Spartan-3 FPGAs. It is an open-source, Harvard
architecture processor and was designed for the high-throughput, data-processing
tasks of the Open Video Graphics Adapter (OpenVGA) project~\cite{OpenVGA}. TTA
processors have very simple instruction word formats and require only very simple
instruction decoders resulting in smaller processor cores. The TTA16 PC circuit
is similar to that shown in Figure~\ref{PC_CIRCUIT} and contains source code for
both types of PC, FSR or radix-2, and can be synthesised with either one.

TTA16 was the processor that showed the greatest frequency improvement with an
FSR PC (see Table~\ref{CPU_Table}), the FSR counter has 22\% greater clock
frequency than with the radix-2 counter. TTA16 configured to use the FSR PC is
substantially faster than with the radix-2 PC because TTA16 was designed to use a
low-latency PC circuit. When TTA16 is synthesised with a radix-2 counter the PC
circuit becomes the critical path limiting maximum frequency. We speculate that
adding an additional pipeline stage to the radix-2 PC circuit may improve maximum
processor frequency, but this would also increase branch latency and FPGA
resources required.

\subsection{RISC16: A Small 16-bit RISC Processor}

RISC16 was also designed for OpenVGA~\cite{OpenVGA}, for comparison with TTA16,
and shares many design elements. The RISC16 core is has five pipeline stages
arranged so that each has similar latency. Decreasing the latency of one
small component, the PC circuit, would not be expected to have a big effect on
overall performance. Table~\ref{CPU_Table} shows that there were no substantial
difference between the FSR and radix-2 PC implementations of the RISC16 processor
core.

\section{Discussion}

Radix-2 counters are the conventional counters used for the processors PC. These
have been studied and used extensively. Memories, including caches, often support
linear, burst transfers assuming a radix-2 count order. Current software tools,
like assemblers and compilers, assume a radix-2 count order as well. FPGAs, like
the Xilinx Spartan-3 and Virtex-5 families, contain carry-chains to support
radix-2 counters as well.

Traditional assemblers will generate object code suitable only for the radix-2 PC
increment sequence. For processors with FSR-based PCs, an assembler is needed
that will use a description of the FSR to correctly re-order the processor
instruction. An assembler was developed for this purpose and reads an XML-encoded
complete processor description prior to generating the assembly output. The
assembler supports FSR and radix-2, as well as hybrid PCs consisting of a
concatenation of an arbitrary number of FSR and radix2 counters. For each FSR
counter used the tap sequence describing the FSR is also encoded within the
processor description file.

\subsection{FSR Program Counters}

The FSR counters are very fast (see Figure~\ref{PC_CIRCUIT}) leading to a
substantially lower-latency PC circuit (see Figure~\ref{GRAPH_COUNTERS}). This in
turn leads to notable gains in maximum processor operating frequency for the
processor with the highest operating frequency, TTA16 (see
Table~\ref{CPU_Table}). Because a maximal-cycle FSR count cycle is different from
a radix-2 counter this has effects on tools, instruction encoding, and cache
coherency.

\subsubsection{Relative Instruction Addressing}

The instruction set of many processors contain branch instructions that store an
offset that is added to the current value of the PC.  These instructions are
called PC-relative branch instructions. When the offset is encoded with a
bit-width narrower than the PC, this branching is difficult to do with FSRs.

PC-relative branching with small offsets is widely used in contemporary
processors~\cite{parhami2005cam} as it allows a subset of addresses, those that
are close to the current instruction, to be encoded within an instruction. The
simplest approach with FSR-based PCs is to not support PC relative branching.
This problem is more easily resolved using Hybrid PCs (see
Section~\ref{HYBRID_CACHE_COHERENCY}).

\subsubsection{Cache Coherency}
\label{MFSR_CACHE_COHERENCY}

For processor designs featuring an instruction cache, a FSR-based PC will lead to
poor performance. A single FSR-based PC circuit will traverse program memory in a
pseudo-random count cycle. Caches are typically designed to fetch multiple words
from sequential adresses, called cache lines. A FSR PC circuit may only execute
one instruction from this cache line, and then the cache may need to fetch
another. A solution to this problem is introduced in Section~\ref{PC_HYBRID} and
then issues are discussed in Section~\ref{HYBRID_CACHE_COHERENCY}.

\subsection{Hybrid Program Counters}
\label{HYBRID_CACHE_COHERENCY}

The synthesis results, shown in Figures~\ref{GRAPH_PC} and~\ref{GRAPH_PC_VIRTEX},
for the hybrid PCs described in Section~\ref{PC_HYBRID} show that performance is
greater than radix-2 PCs at all tested bit-widths. They are therefore a good
solution to the cache coherency problems of FSRs. Figures~\ref{GRAPH_PC}
and~\ref{GRAPH_PC_VIRTEX} show that hybrid PCs have only slightly higher latency
than pure FSR PCs and with the same constant-latency behaviour with increasing
$N$. Hybrid PCs can also be used to solve problems with relative branching and
position independent code, though this is future work.

With a hybrid-PC, the FSR portion will not increment when its value is zero, so
there is one entire cache line which will not be accessed by the count cycle.
This need not be a disadvantage because the first cache line could be used to
store other information, for example the interrupt vector table.

%

\section{Conclusions}

We have designed PC circuits that can allow some FPGA-based processors to operate
at higher frequencies. This is because FSR-based counters have very low latency,
a depth of just one logic gate, and PC circuits utilising FSRs can have
substantially lower latency when implemented in FPGAs. Due to the constant-time
behaviour, with increasing bit-width $N$, FSR counters have an even greater
advantage, relative to radix-2 counters, when $N$ is large.

For small embedded FPGA processors executing instructions stored in local SRAMs,
the pseudo-random count cycle of FSRs is no significant problem either as long as
the user has the necessary tools to generate code. We have also presented hybrid
PCs to solve the FSR cache-coherency issues for processors that use an
instruction cache to reduce average latency for instruction fetching.

\subsection{Future Work}

There can be many possible maximal-cycle FSRs for a particular bit-width. Some
FSRs, due to each having a slightly different circuit, may be faster for a
specific implementation than other FSRs. Testing was not performed to find the
fastest available FSR circuit for a particular implementation. A future project
might be searching amongst the many possible FSRs to find the lowest latency
circuit.

Modern compilers can generate position independent code making use of relative
branching that is not practicable with a pure FSR-based PC. Further work
exploring hybrid program counters, consisting of three or more smaller counters,
could probably be used to solve the FSR relative addressing problems.

The emphasis of this work is on FPGA-based processors. Due to their very low gate
depth and reduced logic complexity FSR and FSR-radix-2 hybrid PCs may also prove
useful with some very low gate-count, or very high clock frequency processors
realised in Silicon.

\bibliographystyle{plain}
\bibliography{thesis_bib}

\end{document}